\documentclass{jocg}
\usepackage{graphicx,amssymb,amsmath,dsfont}
\usepackage{url}
\usepackage{subfig}
\usepackage{enumitem}
\usepackage{wrapfig}

\usepackage{amsthm}
\theoremstyle{plain}
\newtheorem{theorem}{Theorem}
\newtheorem{lemma}[theorem]{Lemma}

\newtheorem{prop}[theorem]{Proposition}
\theoremstyle{definition}

\def\captionstyle{}
\def\boxcaptionstyle{\raggedright}

\def\maxfigfraction{.6}

\newdimen\figboxmargin
\newdimen\figboxhang
\def\DVIscaling{1}
\def\globalscaling{1}
\def\figuredirectory{./figures}
\let\boxer=\llboxer

\def\missingfigure#1{\hbox{Missing figure #1.ps}}

\catcode `\@=11

\newbox\figurebox

\def\figbox{
\@ifnextchar[{\figboxaux}{\figboxaux[htb]}}

\long\def\figboxaux[#1#2]#3#4#5#6{
\writepict{{#3}{#4}{#5}{#6}}
\setbox\figurebox\hbox{#3}%
\if l#1\tryleftbox{#4}{#5}{#6}%
\else
\if r#1\tryrightbox{#4}{#5}{#6}%
\else
\if *#1\checktwocoloptions#2]{\box\figurebox}{#4*}{#5}{#6}%
\else\tryonecol[#1#2]{#4}{#5}{#6}%
\fi
\fi
\fi\ignorespaces}

\long\def\tryleftbox#1#2#3{
\ifdim\wd\figurebox>\maxfigfraction\columnwidth \tryonecol[htb]{#1}{#2}{#3}%
\else\leftbox{\captionbox{\box\figurebox}{#1}{#2}{#3}}\fi}

\long\def\tryrightbox#1#2#3{
\ifdim\wd\figurebox>\maxfigfraction\columnwidth \tryonecol[htb]{#1}{#2}{#3}%
\else\rightbox{\captionbox{\box\figurebox}{#1}{#2}{#3}}\fi}

\def\checktwocoloptions{
\@ifnextchar]{\floatbox[htb}{\floatbox[}}

\long\def\tryonecol[#1]#2#3#4{
\ifdim\wd\figurebox>\columnwidth \floatbox[#1]{\box\figurebox}{#2*}{#3}{#4}%
\else\floatbox[#1]{\box\figurebox}{#2}{#3}{#4}\fi}

\long\def\floatbox[#1]#2#3#4#5{%
\begin{#3}[#1]
\hbox to \hsize{\hfil#2\hfil}
\captionandlabel{#3}{#4}{#5}
\end{#3}
}

\long\def\captionbox#1#2#3#4{
\setbox\figurebox\hbox{#1}%
\parbox[t]{\wd\figurebox}{%
\bigskip\box\figurebox
\let\captionstyle=\boxcaptionstyle
\captionandlabel{#2}{#3}{#4}
\bigskip
}}

\def\captionandlabel#1#2#3{
\def\testit{#3}%
\ifx\testit\empty\else
\writecapt{{#1}{#2}{#3}}
\captypeunstarred#1*.
\getcaption#3\endc@ption
\def\testit{#2}
\ifx\testit\empty\else\label{#2}\fi
\fi}

\def\captypeunstarred#1*#2.{
\def\@captype{#1}}

\def\getcaption{\@ifnextchar[{\getcaptwo}{\getcapone}}
\long\def\getcapone#1\endc@ption{\caption[#1]%
{\def\baselinestretch{1}\Large\normalsize\captionstyle\ignorespaces #1}}
\long\def\getcaptwo[#1]#2\endc@ption{\caption[#1]%
{\def\baselinestretch{1}\Large\normalsize\captionstyle\ignorespaces #2}}

\newdimen\figboxht
\newcount\figboxn

\newcount\figboxlines
\newdimen\figboxwid

\newif\ifisleftbox

\long\def\leftbox#1{%
\setbox\figurebox\hbox{#1}\global\isleftboxtrue
\startmarginbox
\vadjust{\smash{\rlap{\hskip\hsize\hskip\figboxhang
\llap{\raise.7\baselineskip\box\figurebox\hskip\rightskip}}}}%
\endmarginbox%
}

\long\def\rightbox#1{%
\setbox\figurebox\hbox{#1}\global\isleftboxfalse
\startmarginbox
\smash{\llap{\raise.7\baselineskip\box\figurebox\hskip\figboxmargin}}%
\endmarginbox%
}

\def\startmarginbox{%
\ifvmode\passpict\let\endmarginbox=\indent
\else\message{WARNING: marginbox in not in vmode}\hfilneg\ \passpict
\let\endmarginbox=\relax\fi
\figboxht=\dp\figurebox
\advance\figboxht by 1.3\baselineskip
\vskip.95\figboxht\penalty-300\vskip-.95\figboxht
\divide\figboxht by\baselineskip
\global\figboxlines=\figboxht
\global\figboxwid=\wd\figurebox
\global\advance\figboxwid by \figboxmargin
\global\advance\figboxwid by -\figboxhang
\setmypar\noindent}
\def\addlines#1{\global\advance\figboxlines by #1\myparshape}
\def\zerolines{\origpar\global\figboxlines=0\myparshape}

\def\passpict{\par\ifnum\figboxlines>1\vskip\figboxlines\baselineskip
\zerolines\fi}

\def\emptybox#1#2{\hbox to #1{\vbox to #2{\vss}\hss}}

\global\let\origpar=\@@par
\global\let\dopar=\origpar
\global\def\@@par{\dopar}
\@setpar{\dopar}

\def\setmypar{\global\let\dopar=\mypar
\global\prevgraf=0\myparshape}

\def\mypar{\origpar\global\advance\figboxlines by -\prevgraf%
\global\prevgraf=0\myparshape}

\def\myparshape{\relax%
\ifnum\figboxlines>1\theparshape \else
\global\hangindent=0pt\global\hangafter=1
\global\let\dopar=\origpar\fi}

\def\theparshape{%
\ifisleftbox\global\hangindent=-\figboxwid 
\else\global\hangindent=\figboxwid \fi
\global\hangafter=-\figboxlines \global\advance\hangafter by 1%
}

\def\definefnum#1{
\def\fnum@figure{Figure \ref{#1}}%
\def\fnum@table{Table \ref{#1}}%
\def\fnum@code{Algorithm \ref{#1}}%
}

\def\writepict#1{}
\def\writecapt#1{}

\def\journalpicts#1{
\newwrite\pictfile
\newwrite\captfile
\openout\pictfile\jobname.pic
\openout\captfile\jobname.cap
\gdef\writepict##1{\unexpandedwrite\pictfile{\doit##1}}%
\gdef\writecapt##1{\unexpandedwrite\captfile{\doit##1}}%
\global\let\ENDdocument=\enddocument
\gdef\enddocument{\DOjournalpicts{#1}\ENDdocument}
}

\def\DOjournalpicts#1{{%
\def\writepict##1{}\closeout\pictfile
\def\writecapt##1{}\closeout\captfile
\@fileswfalse
\onecolumn
\def\globalscaling{#1}
\def\doit##1##2##3##4{
\figboxaux[t]{\hss##1\hss}{##2}{}{}%
\vspace*{1in}
\definefnum{##3}
\captionandlabel{##2}{}{##4}
\clearpage}%
\input\jobname.pic
\def\doit##1##2##3{
\definefnum{##2}
\captionandlabel{##1}{}{##3}}%
\raggedright\let\captionstyle=\raggedright
\def\@makecaption##1##2{##1: ##2\par}
\input\jobname.cap
}}

\def\llboxer#1{\vbox to \figboxht{\vfil\hbox to \figboxwid{#1\hfill}}}
\def\lcboxer#1{\vbox to \figboxht{\vfil\hbox to \figboxwid{\hfill#1\hfill}}}
\def\oldboxer#1{\vbox to \figboxht{\vfil
                      \hbox to \figboxwid{\hfill\llap{#1\hskip4.25in}\hfill}}}
\def\ulboxer#1{\vbox to \figboxht{\hbox to \figboxwid{#1\hfill}\vfil}}
\def\ccboxer#1{\vbox to \figboxht{\vfil
                        \hbox to \figboxwid{\hfill#1\hfill}\vfil}}

{\catcode`\p=12\catcode`\t=12
\gdef\removedimen#1pt{#1}}

\def\defscaled#1#2{#2=\DVIscaling#2%
\xdef#1{\expandafter\removedimen\the#2}}

\def\DVIspace{ }
\newdimen\hscalefactor
\newdimen\vscalefactor

\def\scale#1{\horizscale{#1}\vertscale{#1}}
\def\horizscale#1{\hscalefactor=#1\hscalefactor\figboxht=#1\figboxht}
\def\vertscale#1{\vscalefactor=#1\vscalefactor\figboxwid=#1\figboxwid}

\def\boxps{%
\@ifnextchar[{\boxpsaux}{\boxpsaux[\relax]}}

\def\boxpsaux[#1]#2#3#4#5{%
{\figboxwid#4\figboxht#5\hscalefactor=1pt\vscalefactor=1pt%
\scale{#3}%
\scale{\globalscaling}%
#1%
\defscaled\DVIhscale\hscalefactor
\defscaled\DVIvscale\vscalefactor
\boxer{\includegraphics{\figpsfilename\DVIspace}}}%
}

\newread\Epsffilein
\newif\ifEpsffileok
\newif\ifEpsfbbfound

\newdimen\pspoints
\divide\pspoints by 72

\def\boxeps{%
\@ifnextchar[{\boxepsaux}{\boxepsaux[\relax]}}
\def\boxepsaux[#1]#2#3{%
\gdef\figpsfilename{\figuredirectory/#2.ps}
\openin\Epsffilein=\figuredirectory/#2.ps
\ifeof\Epsffilein 
\gdef\figpsfilename{\figuredirectory/#2.eps}
\openin\Epsffilein=\figuredirectory/#2.eps \fi
\ifeof\Epsffilein\message{I couldn't open \figuredirectory/#2.ps or \figpsfilename}%
\missingfigure{#2}
\else
   {\Epsffileoktrue\Epsfbbfoundfalse
    \catcode`\%=11 \catcode`\\=11
    \catcode`\{=11 \catcode`\}=11
    \catcode`\$=11 \catcode`\^=11
    \catcode`\&=11 \catcode`\#=11
    \catcode`\~=11 \catcode`\_=11
    \loop
       \read\Epsffilein to \Epsffileline
       \ifeof\Epsffilein\Epsffileokfalse\else
          \expandafter\Epsfaux\Epsffileline . .\\%
       \fi
   \ifEpsffileok\repeat
   \ifEpsfbbfound
        \figboxht=\Epsfury\pspoints
        \advance\figboxht by-\Epsflly\pspoints
        \figboxwid=\Epsfurx\pspoints
        \advance\figboxwid by-\Epsfllx\pspoints
   \else
        \message{No bounding box comment in \figpsfilename }%
        \figboxwid=2in\figboxht=1in%
   \fi%
   \immediate\closein\Epsffilein
   \hscalefactor=1pt\vscalefactor=1pt%
   \scale{#3}%
   \scale{\globalscaling}%
   #1%
   \defscaled\DVIhscale\hscalefactor
   \defscaled\DVIvscale\vscalefactor
   \hscalefactor=-\Epsfllx\hscalefactor
   \hscalefactor=1.00375\hscalefactor
   \defscaled\DVIhoffset\hscalefactor
   \vscalefactor=-\Epsflly\vscalefactor
   \vscalefactor=1.00375\vscalefactor
   \defscaled\DVIvoffset\vscalefactor
   \llboxer{\includegraphics{\figpsfilename\DVIspace}} }%
\fi
}%
{\catcode`\%=11 \global\let\Epsfpar=\par
\global\let\Epsfpercent=
\long\def\Epsfaux#1#2 #3\\{\relax\ifx#1\Epsfpercent
   \def\testit{#2}\ifx\testit\Epsfbblit
      \Epsfsize #3 . . . .\\%
      \global\Epsffileokfalse
      \global\Epsfbbfoundtrue
   \fi\else\ifx#1\Epsfpar\else\global\Epsffileokfalse\fi\fi}%
\def\Epsfsize#1 #2 #3 #4 #5\\{\global\def\Epsfllx{#1}\global\def\Epsflly{#2}%
   \global\def\Epsfurx{#3}\global\def\Epsfury{#4}}%

\catcode`\@=12                  

\def\pic#1;#2;#3;#4\par{\picsc#1;#2;#3;1;#4\par}

\def\picsc#1;#2;#3;#4;#5\par{
\figbox[htb]{\boxeps{#1}{#4}
}{figure}{#1}{%
#5}}

\def\mpic#1;#2;#3;#4\par{\mpicsc#1;#2;#3;1;#4\par}

\def\mpicsc#1;#2;#3;#4;#5\par{
\figbox[l]{\boxeps{#1}{#4}
}{figure}{#1}{%
#5}}

\newtheorem{definition}[theorem]{Definition}
\def\R{\mathds{R}}
\let\epsilon=\varepsilon

\def\tomath#1{\relax\ifmmode#1\else$#1$\fi}

\def\fref#1{Figure~\ref{fig:#1}}
\def\sref#1{Section~\ref{sec:#1}}

\def\lref#1{Lemma~\ref{lem:#1}}
\def\tref#1{Theorem~\ref{thm:#1}}

\def\ray#1{\hbox{\vbox{\offinterlineskip\setbox0\hbox{$#1$}
	\hbox to \wd0{\hss$\rightharpoonup$\hss}\vskip-1.0pt\box0}}}
\def\lin#1{\hbox{\vbox{\offinterlineskip\setbox0\hbox{$#1$}
	\hbox to \wd0{\hss$\longleftrightarrow$\hss}\vskip-1.0pt\box0}}}

\long\def\comm#1{\ignorespaces}
\def\comments{\long\def\comm##1{\message{COMMENT: ##1}{\bf(( ##1 ))}}}

\long\def\pa#1{\ignorespaces}

\comments

\title{\MakeUppercase{Covering Folded Shapes}\thanks{A preliminary
    version of this work was presented at the 25th Canadian Conference
    on Computational Geometry (CCCG'13)~\cite{aaddfhlsw-cfs-13}.}}

\date{}

\index{Aichholzer, Oswin}
\index{Aloupis, Greg}
\index{Demaine, Erik D.}
\index{Demaine, Martin L.}
\index{Fekete, S\'andor P.}
\index{Hoffmann, Michael}
\index{Lubiw, Anna}
\index{Snoeyink, Jack}
\index{Winslow, Andrew}

\author{
  Oswin Aichholzer,%
  \thanks{\affil{Institute for Software Technology, TU Graz, Inffeldgasse 16b/II,
      A-8010 Graz, Austria},
    \email{oaich @ist.tugraz.at}. Partially supported by the ESF EUROCORES
    programme EuroGIGA---CRP `ComPoSe', Austrian Science Fund (FWF):
    I648-N18.}\,
  Greg Aloupis,%
  \thanks{\affil{Department of Computer Science, Tufts Univ.,
      161 College Ave., Medford, MA 02155, USA},
    \email{aloupis.greg@gmail.com},\email{awinslow@cs.tufts.edu}. 
      Partially supported by NSF grant CBET-0941538.}\,
  Erik D. Demaine,%
    \thanks{\affil{Computer Science and Artificial Intelligence Laboratory, MIT,
      32 Vassar St., Cambridge, MA 02139, USA},
      \email{{edemaine,mdemaine}@mit.edu}.}\,
  Martin L. Demaine,%
  \footnotemark[4]\,
  S\'andor P. Fekete,%
  \thanks{\affil{Department of Computer Science,
      TU Braunschweig,
      M\"uhlenpfordtstr.~23, 38106 Braunschweig, Germany},
    \email{s.fekete@tu-bs.de}}\,
 Michael Hoffmann,%
 \thanks{\affil{Department of Computer Science, ETH Z\"urich,
     Switzerland},
   \email{hoffmann@inf.ethz.ch}. Partially supported by the ESF EUROCORES
   programme EuroGIGA, CRP GraDR and 
   SNF Project 20GG21-134306.}\,
 Anna Lubiw,%
 \thanks{\affil{David R. Cheriton School of Computer Science,
     Univ.~Waterloo, Waterloo, ONT N2L 3G1, Canada},
 \email{alubiw@uwaterloo.ca}}\,
 Jack Snoeyink,%
 \thanks{\affil{Department of Computer Science, University of North Carolina,
     Chapel Hill, NC 27599, USA},
   \email{snoeyink@cs.unc.edu}. Partially supported by an NSF grant.}\,
 Andrew Winslow%
 \footnotemark[3]
}

\begin{document}
\pagestyle{fancy}
\maketitle
\begin{abstract}
Can folding a piece of paper flat make it larger? 
We explore whether a shape $S$ must be scaled to cover a flat-folded copy
of itself. We consider both single folds and arbitrary folds (continuous
piecewise isometries $S\to \R^2$). The underlying problem is motivated
by computational origami, and is related to other covering and fixturing
problems, such as Lebesgue's universal cover problem and force closure
grasps. In addition to considering special shapes (squares, equilateral
triangles, polygons and disks), we give upper and lower bounds on scale
factors for single folds of convex objects and arbitrary folds of
simply connected objects.
\end{abstract}


\section{Introduction}
\label{sec:intro}

\begin{wrapfigure}[8]{r}{6cm}
  \centering
  \vspace*{-4mm}
  \includegraphics[width=4.0cm]{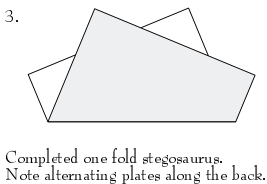}
  \caption{From Wu's diagram.}
  \vspace{-10pt}
  \label{fig:stego}
\end{wrapfigure}
We explore how folds can make an origami model larger, in the sense that 
Joseph Wu's one-fold stegosaurus\footnote{An origami joke. 
\url{http://www.josephwu.com/Files/PDF/stegosaurus.pdf}}
cannot be covered by a copy of the square from which it is folded.
In more technical terms, we consider how to cover all possible folded versions of a given shape
by a scaled copy of the shape itself, with the objective of keeping the scale
factor as small as possible.  



Problems of covering a family of shapes by one minimum-cost object
have a long tradition in geometry. The classical prototype is
Lebesgue's universal cover problem from 1914~\cite{p-uev-20}, which
asks for a planar convex set of minimum area that can cover any planar
set of diameter at most one; Brass and Sharifi~\cite{bs-lbluc-05} give
upper and lower bounds, but a gap remains.  A
similar question, also with a gap, is Moser's worm
problem~\cite{npl-wplm-92,pww-cnsua-07}, which asks for a convex set
of minimum area that can cover any planar curve of unit length.  
As reported in 
the book by
Brass, Moser, and Pach~\cite[Chapter 11.4]{bmp-rpdg-05}, there is a large family of
well-studied, but notoriously difficult problems parameterized by
\begin{itemize}
\item the family of sets to be covered,
\item the sets allowed as covers,
\item the size measure to be minimized, and
\item the allowed transformations.
\end{itemize}

\smallbreak
In this paper we consider a given {\it shape $S$}, which is a bounded region of the plane that is a
simply connected (no holes)
closed 2-manifold with boundary (every interior point has a disk
neighborhood and every boundary point  a half-disk).
A shape $S$ may  possess more specific properties:
e.g., it may be convex, a (convex or non-convex) polygon, a disk, a square, or an
equilateral triangle.  

We denote by $cS$, for $c>0$, the family of copies of $S$ that
have been scaled by $c$, and then rotated, reflected, and translated.  
We consider
upper and lower bounds on the smallest constant $c$ such that, for
any~$F$ obtained by folding~$S$, some member of $cS$ contains or {\it
  covers}~$F$.  Let us be more specific about folding.

A {\it single fold of $S$} with line $\ell$
reflects one or more connected components of the difference $S\setminus \ell$ across~$\ell$. Let ${\cal F}_1(S)$ denote the family of shapes that can
be generated by 
a single fold of $S$. 
An {\it arbitrary fold of $S$} is a continuous,
piecewise isometry from $S\to \R^2$, which partitions $S$ into a
finite number of polygons and maps each rigidly to the plane so that
the images of shared boundary points agree. The key property of
arbitrary folds 
is that the length of any path in~$S$ equals the length of its image
in~$\R^2$.  Let ${\cal F}(S)$ denote the family of shapes that can
be generated by an arbitrary fold of $S$.  

The single fold and arbitrary fold are two simple notions of
flat folding that avoid concerns of layering and fold order. 
Note that any upper bound that we prove for arbitrary
folds applies to single folds, too.  And, although the image of an
arbitrary fold need not be the result of single folds, our
lower bounds happen to be limits of finite sequences of single folds.
Our results  apply to 3-d folded shapes if 
{\it covering} is understood to mean covering the orthogonal
projection to the plane.

Throughout this paper, we consider 
the following type of covers:
\begin{definition}
\label{def:cover}
For a given shape $S$ and $c>0$, $c S$ is an {\em origami cover} of $S$ 
if any member of ${\cal F}(S)$  can be covered by
some member of $c S$.  The {\em
origami cover factor} $c^*(S)$ is the smallest such $c$, which may be
$\infty$:
\[
c^*(S)=\inf \{c\mid \hbox{$cS$ is an origami cover of $S$}\}.
\]

Analogously, $c S$ is a {\em $1$-fold cover} of $S$ if any member of
${\cal F}_1(S)$ can be covered by some member of $c S$; and the {\em
  $1$-fold cover factor} $c_1^*(S)$  is the smallest such $c$:
\[
c_1^*(S)=\inf \{c\mid \hbox{$cS$ is a 1-fold cover of $S$}\}.
\]
\end{definition}

Note that by definition $\mathcal{F}_1(S)\subseteq\mathcal{F}(S)$ and
so $c_1^*(S)\le c^*(S)$, for any shape $S$.

Questions of whether folding can increase area or perimeter  
have been considered before.
It is clear that folding a piece of paper introduces overlap, so area can
only decrease.  On the other hand, the perimeter of a rectangle or square can
be greater in a folded than an unfolded state---known 
as Arnold's ruble note or the Margulis napkin problem~\cite{a-ap-05,l-odsmmaa-03}. 
Folding techniques that increase perimeter, like rumpling and
pleat-sinking, make very small but spiky  models that are easily
covered by the original paper shape, however.

Let us recall some common geometric parameters of shapes and derive a
first simple general upper bound for the origami cover factor in terms
of these parameters. For a given shape $S$, an {\it incircle}, $C_r$,
is a circle of maximum radius (the {\it inradius}~$r$) contained
in~$S$.  Similarly, the {\it circumcircle}, $C_R$, is the circle of
minimum radius (the {\it circumradius}) that contains~$S$. In order to
extend these notions to non-convex shapes, we consider geodesic
distances, that is,~the distance between two points in $S$ is the
length of a shortest path that connects the points and stays within
$S$. The maximum geodesic distance $D$ between any two points of~$S$ is the
{\it geodesic diameter} of $S$.
A {\it geodesic center} is a point in $S$ that minimizes the maximum
distance (the {\it geodesic radius}~$R$) to all points of~$S$. For
convex shapes the geodesic radius $R$ is also the circumradius. Jung's
theorem in the plane says $\sqrt 3 R\le D\le 2R$, with the equilateral
triangle and circle giving the two
extremes~\cite[ch.~16]{rademacher1990enjoyment}.  For any
shape 
$S$, these parameters give an upper bound on the origami cover factor.
\begin{lemma}\label{lem:ub}
  Any shape $S$ with inradius $r$ and geodesic radius $R$ has an
  origami cover factor \[c^*(S)\le R/r.\]
\end{lemma}
\begin{proof}
Place any folded state $F\in {\cal F}(S)$ in the plane so that the image of a 
geodesic center is at the origin.  Choose a member of $(R/r) S$ with
an incircle center at the origin.  Because no path in $F$ can be more
than~$R$ from the origin, the scaled incircle covers~$F$. 
\end{proof}

There are shapes for which the bound of \lref{ub} is tight. For
instance, if $S$ is a disk, then $r=R$ and $c^*(S)=1$.




\paragraph{Results.} The remainder of the paper is divided into two
parts. First, in \sref{single} we consider single folds and present
bounds for the $1$-fold cover factor of various families of shapes:
\begin{itemize}
\item If $S$ is a convex shape with inradius $r$ and circumradius $R$,
  then $c^*(S)\ge c_1^*(S)\ge\kappa R/r$, where
  $\kappa=((\sqrt{5}-1)/2)^{5/2}\approx 0.300283$. Note that this
  bound is within a constant factor of the general upper bound from
  \lref{ub}.
\item If $S$ is an equilateral triangle, then $c_1^*(S)=4/3$.
\item If $S$ is a square, then $c_1^*(S)=\varrho\approx 1.105224$,
  where $\varrho$ denotes the largest (and only positive) real root of
  the polynomial
  $\Phi(x)=40x^{12}+508x^{11}+1071x^{10}+930x^9-265x^8-1464x^7-1450x^6-524x^5+58x^4+76x^3+3x^2-6x-1$.
\item If $S$ is a polygon, then $c_1^*(S)>1$, that is, for any polygon
  $S$ there is a single fold such that the resulting folded state
  cannot be covered with a copy of $S$.
\item On the other hand, we describe an infinite family of shapes for
  which $c_1^*(S)=1$; these are shapes cut from a disk.
\end{itemize}

Then in \sref{arb} we discuss arbitrary folds and present a lower
bound for the origami cover factor for a more general family of
shapes: 
\begin{itemize}
\item For a simply connected shape $S$ with inradius $r$ and
geodesic radius $R$, we have $c^*(S)\ge\kappa R/r$, where
$\kappa=\sqrt{3}/(2\pi)\approx 0.27566$. 
\item We describe a family
of shapes that have an origami cover factor of $1$, like disks. In
fact, these shapes are constructed as a union of two disks.
\end{itemize}




\section{Single Folds}\label{sec:single}
In this section we explore the 1-fold cover factor $c_1^*(S)$, giving general bounds for convex $S$ and for polygons, and  the exact values for equilateral triangles, squares, and a family derived from disks.

\subsection{Convex Shapes}
\label{sec:convex}

For a convex set $S$, there is a lower bound
for the 1-fold cover factor $c_1^*(S)$ that is within a constant factor of the
upper bound given by \lref{ub}.
\begin{theorem}
\label{thm:convex}
Let $S$ be a convex shape with inradius $r$ and circumradius $R$.
Then $\kappa R/r\leq c^*(S)\leq R/r$ for
an appropriate constant $\kappa=((\sqrt{5}-1)/2)^{5/2}\approx 0.300283$. 
\end{theorem}

\begin{proof}
The upper bound is from \lref{ub}. 

For the lower bound, consider the center $p^*$ of the circumcircle $C_R$ that contains $S$. Because 
$R$ is smallest possible, the set of points where the boundary of $C_R$ touches~$S$, $T:=\partial C_R\cap S$, must
contain at least two points, and no open halfplane through $p^*$ can contain all of~$T$. 
If $|T|=2$, then these two points $t_1$ and $t_2$ must lie on a diameter of $C_R$;
if $|T|>2$, there must be two points $t_1, t_2\in T$ that form a  
central angle $\angle(t_1,p^*,t_2)$ in $[\frac{2}{3}\pi,\pi]$. 
Thus, for any $\varphi\in [0, \frac{2}{3}\pi]$, we can
perform a single fold along a line through $p^*$ that maps $t_2$ to $t_2'$ such that the central angle
$\angle(t_1,p^*,t_2')$ is $\varphi$.  

\begin{figure}[ht]
\centering
\includegraphics
{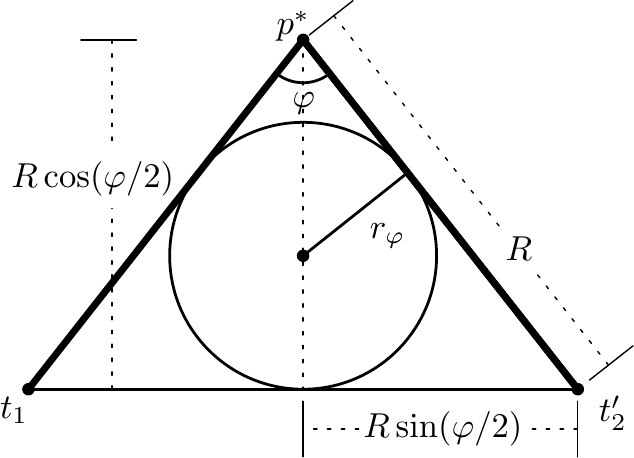}
\caption{Parameters for calculating the 1-fold cover factor for convex~$S$.}
\label{fig:convex}
\end{figure}

Now, after folding, consider a cover of the three points $t_1,p^*,t_2'$ by $c S$ for some $c>0$. As each member of $c S$ 
is convex, in covering the  triangle $\Delta(t_1,p^*,t_2')$, it also covers the largest circle $C_\Delta$ contained
in $\Delta(t_1,p^*,t_2')$; let $r_\varphi$ be the radius of this
circle, see \fref{convex}. 
Using elementary geometry we obtain $r_\varphi=\frac{R}{2}\frac{\sin(\varphi)}{1+\sin(\varphi/2)}$,
which is maximized at
$\varphi=2\arctan\bigl(((\sqrt{5}-1)/2)^{1/2}\bigr) \approx
76.345^\circ$, giving $r_\varphi=\kappa R$ as the radius of
$C_\Delta$.  Because the largest circle covered by a member of $cS$
has radius~$cr$, and $C_\Delta$ is covered by $c S$, we conclude that
$c\geq \kappa R/r$.
\end{proof}


\subsection{Cover Factors for Specific Polygons}
\label{sec:specific}
In this section we determine $c_1^*(S)$ when $S$ is an
equilateral triangle or a square. These two cases illustrate analysis techniques that could in theory be extended to other polygons, except that the number of cases explodes, especially for non-convex shapes.

{\begin{figure}[ht] 
\centering
\subfloat[Folded equilateral triangle and two minimum
enclosing triangles with $c_1^*=4/3$.]
{\hspace*{10pt}\includegraphics
{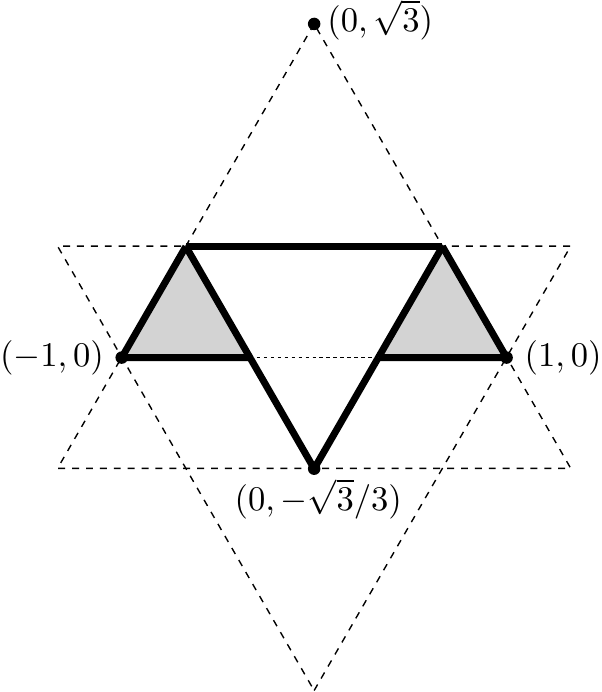}\hspace*{10pt}
\label{fig:triangle}}
\hfil%
\subfloat[Folded square and three minimum
enclosing squares with $c_1^*\approx1.105$.]
{\hspace*{10pt}\includegraphics
{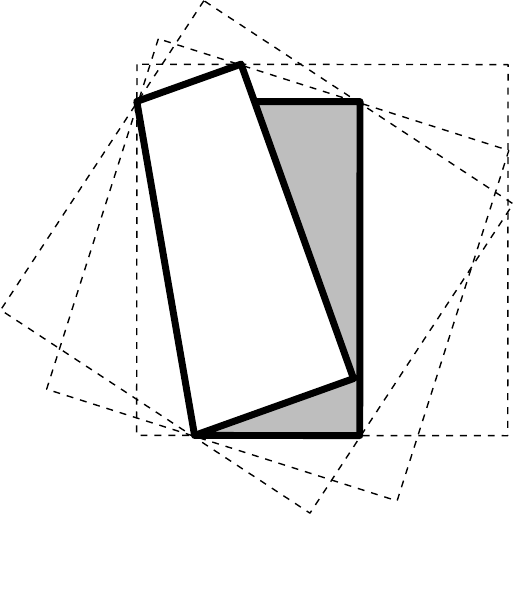}\hspace*{10pt}
\label{fig:3enclosingsq}}
\caption{Optimal $1$-fold covers for equilateral triangle and square.}
\end{figure}

An important subproblem is to fix the folded shape $F$ and compute,
for the given shape~$S$, the smallest $c$ such that $cS$
covers~$F$. With four degrees of freedom for translation, rotation,
and scaling, we expect that four first-order contacts between the
boundaries of $S$ and $F$ will define the minimum $c$. In polygons,
these will be four pairs consisting of a vertex $v$ of $F$ and an edge
$e$ of $S$ such that $v$ lies on $e$.


For equilateral triangles, we can use the following reformulation of a
lemma by DePano and Aggarwal:
\begin{lemma}[{\cite[Lemma 2]{da-frecp-84}}]
\label{lem:MET-characterization}
The smallest enclosing equilateral triangle of a polygonal shape has
at least one vertex of the shape on each side and at least one side of
the triangle contains two vertices of the shape. (A vertex of the
shape in the corner of the triangle counts for both incident sides.)
\end{lemma}

When looking for an enclosing square, there is an additional possibility~\cite{dgn-skersao-05,da-frecp-84}; the minimum may have four points in contact with four different sides. These papers compute such minima by solving for roots of polynomials, but an appealing direct construction 
of the square through four
points, which is unique when it exists, is in Problem~20 in Kovanova and
Radul's list of ``Jewish problems''~\cite{kr-jp-11}:  for points $A$--$D$ in ccw order, construct $BD'$ perpendicular and of equal length to $AC$; If $D'\ne D$, then two sides of the square must be parallel to $DD'$.

We use the following lemma, which can be found phrased slightly
differently in Das et~al.~\cite{dgn-skersao-05}:
\begin{lemma}[\cite{dgn-skersao-05}]
\label{lem:square-boundary}
  For any compact set $P\subset\R^2$ there exists a smallest enclosing
  square $S$ of $P$ that is of one of the following two combinatorial
  types:
  \begin{enumerate}[label=(\arabic{*})]
  \item\label{p:square:1} each side of $S$ contains a point from $P$;
  \item\label{p:square:2} one side of $S$ contains at least two points
    from $P$, and both the opposite side and an adjacent side of $S$
    each contain at least one point from $P$.
  \end{enumerate}
\end{lemma}

These structural characterizations support the use of rotating
calipers (see e.g.~\cite{t-sgprc-83}) to compute minimum enclosing
shapes. 
In what follows we show that the folds that define $c_1^*(S)$ 
are characterized by having multiple equal-sized enclosing shapes.

\subsubsection{Equilateral Triangle}
The example that establishes the maximum $1$-fold cover factor of an equilateral triangle is nicely symmetric.
\begin{theorem}
\label{lem:triangle-upper-bound}
The $1$-fold cover factor of an equilateral triangle, $c_1^*(\triangle)$, is $4/3$.
\end{theorem}
\begin{proof}
Let $S$ be the triangle of side length 2 with vertices $(\pm1, 0)$ and $(0, \sqrt{3})$.
We begin by showing that any single fold can be covered by scaling to at most~$4/3$. 

By symmetry, we may assume that we fold along a line $y=mx+b$ that intersects both edges incident on $(0, \sqrt{3})$; let $P$ be the image of this vertex in the folded state $S' \in \mathcal{F}_1(S)$. 
For example, in Figure~\ref{fig:triangle}, $P = (0, -\sqrt{3}/3)$.
Consider three cases for the location of the image $P$ and the resulting minimum enclosing equilateral triangle, depicted in Figure~\ref{fig:trianglecases}.

\figbox{\includegraphics%
  {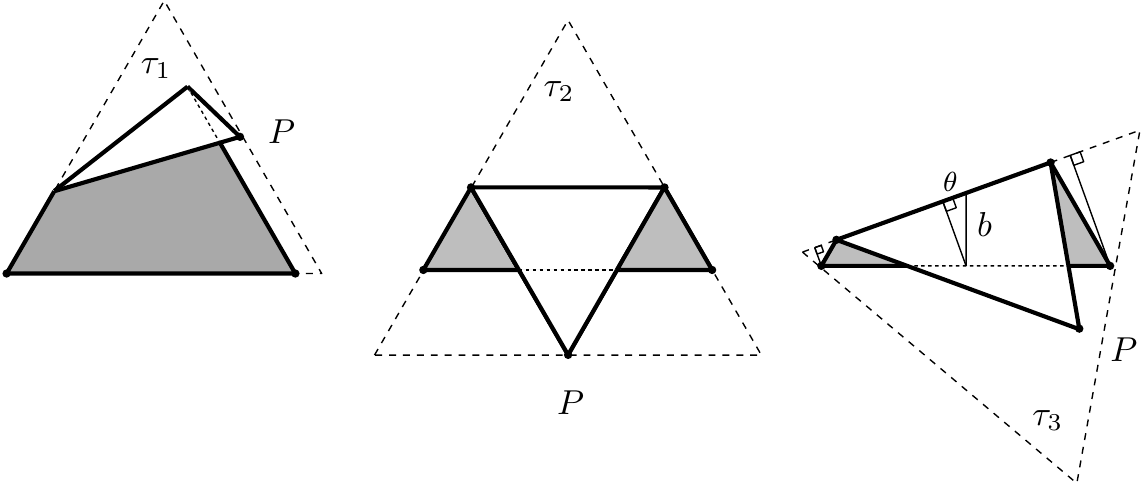}}{figure} {fig:trianglecases} {Cases for
  enclosing triangles depending on $P$. Point $P\in \tau_3$ should be
  below $P\in\tau_2$, but then the small triangles mentioned in the
  proof are hard to see.}

First, suppose that $P$ is on or above the $x$-axis.  By symmetry, we may assume that $P$ lies in the wedge formed by extending both edges of $S$ incident on vertex $(-1,0)$ to rays from~$(-1,0)$.  Because
$P$ has distance at most $2$ from  $(-1, 0)$, 
 scaling $S$ about  $(-1, 0)$ by $2/\sqrt{3}< 4/3$  creates an enclosing equilateral triangle~$\tau_1$.  

Second, suppose that the image $P=(p_x,p_y)$ has $-\sqrt3/3\le p_y\le 0$.
Consider the enclosing triangle $\tau_2$ obtained by scaling $S$
about $(0,\sqrt3)$ until the horizontal edge touches $P$. The scale
factor for this triangle is $ \frac{\sqrt3-p_y}{\sqrt3}= 1-p_y/\sqrt{3}\le 4/3$.

Finally, suppose that $P=(p_x,p_y)$ has $p_y \le-\sqrt3/3$.  
From the previous case, the scale factor for enclosing triangle $\tau_2$ is $1-p_y/\sqrt3 \ge 4/3$.
So instead consider an enclosing triangle $\tau_3$ with an edge $e$ along  the 
fold line, which we can parameterize by its $y$-intercept $b\le \sqrt3/3$ and angle from horizontal $\theta$. 
Draw perpendiculars to $e$ through vertices
$(\pm1,0)$ to form two small 30-60-90 triangles.  Edge $e$ is composed of the short sides of these triangles plus the projection of the base edge of $S$, so $e$ has length $(2+2b/\sqrt3)\cos \theta$.  Thus, the scale factor of triangle $\tau_3$ is $(1+b/\sqrt 3)\cos\theta\le 4/3\cos\theta\le 4/3$.

These cases show that $c_1^*(\triangle) \le 4/3$, and also reveal necessary conditions for equality: the fold line angle $\theta=0$ and intercept $b=\sqrt3/3$, so $P=(0,-\sqrt3/3)$.  To show that these are sufficient, we must check one more candidate for enclosing triangle.  


\begin{figure}[ht]
\centering
\includegraphics
{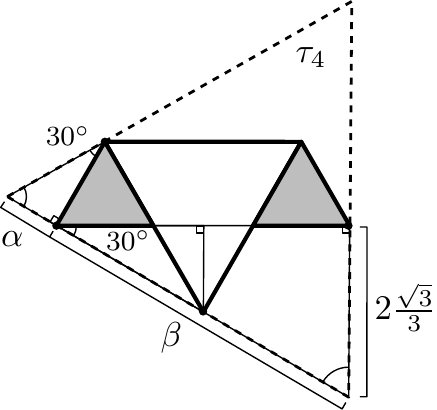}
\caption{Not a minimum enclosing triangle.}
\label{fig:t4}
\end{figure}

Consider $\tau_4$, with edge incident to $P=(0,
-\sqrt{3}/3)$ and $(-1, 0)$. The length of this edge is the sum of sides of two
30-60-90 triangles, marked $\alpha$ and $\beta$ in Figure~\ref{fig:t4}.
The scale factor $(\alpha+\beta)/2 = \sqrt{3}/9+ 2\sqrt{3}/3 =
7\sqrt{3}/9 > 4/3$. Thus, $\tau_4$ is \emph{not} a minimum enclosing
triangle, and $c_1^*(\triangle) = 4/3$, as determined by~$\tau_2$ and~$\tau_3$.
\end{proof}

\subsubsection{Square}

%

For squares, the 
fold that realizes the maximum 1-fold cover factor is astonishingly
complex, and is neither symmetric, nor rational.  For the unit square
$[0,1]^2$, the vertex $(0, 1)$ folds to a location whose $y$
coordinate is the root of a degree twelve polynomial:
$\Phi(x)=40x^{12}+508x^{11}+1071x^{10}+930x^9-265x^8-1464x^7-1450x^6-524x^5+58x^4+76x^3+3x^2-6x-1$.
This polynomial will arise because the optimal fold has \emph{three}
distinct minimum enclosing squares.  Let $\varrho$ denote the largest
(and only positive) real root of $\Phi(x)$, which is approximately
$1.105224$.

%
Let $S=\{(x,y):0\le x,y\le 1\}$ denote the axis-parallel unit square
and consider some $F\in\mathcal{F}_1(S)$ such that $F\ne S$. Note that
$F$ is a simple polygon that is uniquely determined (up to symmetry)
by a fold line $\ell$. 
\begin{prop}\label{prop:square-folded-polygon}
  The polygon $F$ can be covered by $S$, unless fold line $\ell$ intersects $S$
  in the relative interior of two opposite sides.
\end{prop}
\begin{proof}
  If $\ell$ does not intersect the interior of $S$ then $F\cong
  S$. Otherwise $\ell$ intersects $\partial S$ in exactly two
  points. If these points lie on adjacent sides of $S$, then folding along $\ell$ reflects  the triangle formed by these sides and $\ell$ inside the portion of the square $S$ on the opposite side of $\ell$. Therefore, $F$ can be covered by $S$.
\end{proof}
We are interested in a fold line $\ell$ that maximizes the smallest
enclosing square of $F$. Using symmetry with 
Proposition~\ref{prop:square-folded-polygon}, we  can assume:
\begin{enumerate}[label=(\arabic{*})]
\item the line $\ell$ intersects both horizontal sides of $S$ (else
  rotate by $90^\circ$);
\item the slope of $\ell$ is negative (else reflect vertically);
\item $\ell$ intersects the top side of $S$ left of the midpoint~$(1/2,1)$ (else
  rotate by $180^\circ$).
\end{enumerate}
If we imagine $F$ as the result of folding the part of $S$ to the left
of $\ell$ over to the right, then we can parameterize $\ell$ by the
image $P=(p_x,p_y)$ of the top left corner $(0,1)$ of $S$ under this
fold. Under the above assumptions,  a line $\ell$ that passes
(almost) through $(1/2,1)$ and $(1,0)$ would maximize 
$p_y$. Therefore $0<p_x<4/5$ and so
$1<p_y<\sqrt{2p_x-{p_x}^2+1}<7/5$.

Denote the  two points of intersection between $\ell$ and $\partial
S$ by $B=(b_x,0)$ and $T=(t_x,1)$ and denote the image of the
bottom-left corner $(0,0)$ of $S$ under the fold across $\ell$ by
$Q=(q_x,q_y)$. If $q_x>1$, then the convex hull $\mathcal{CH}(F)$ of
$F$ is the hexagon $B,(1,0),Q,(1,1),P,T$, else $Q$ does not appear on
$\partial(\mathcal{CH}(F))$ and it is only a pentagon. Note that in any
case the width of $F$ in the $y$-direction is greater than one, whereas
the width in the $x$-direction is less than one.

For a given $P=(p_x,p_y)$, we have
\begin{align*}
 \ell : y&=-\frac{p_x}{p_y-1}x+\frac{{p_x}^2+{p_y}^2-1}{2(p_y-1)},\\
   T &= \Bigl(\frac{{p_x}^2+(p_y-1)^2}{2p_x},1\Bigr),\\
  B &= \Bigl(\frac{{p_x}^2+{p_y}^2-1}{2p_x},0\Bigr),\,\mathrm{and}\\
  Q &= \Bigl(\frac{p_x({p_x}^2+{p_y}^2-1)}{{p_x}^2+(p_y-1)^2},\frac{({p_x}^2+{p_y}^2-1)(p_y-1)}{{p_x}^2+(p_y-1)^2}\Bigr).
\end{align*}

What does a smallest enclosing square $\sigma$ of $F$ look like?  For
the upper bound on the cover factor we consider three enclosing
squares (\figurename~\ref{fig:square-lower-bound-candidates}).
\begin{description}
\item[$\sigma_1$] is the smallest axis-parallel enclosing square, which has points $B$ and $(1,0)$ on the bottom side, $P$ on the top, $T$ on the left, and no point on the right.
\item[$\sigma_2$] has 
  points $P$ and $(1,1)$ on one side, $B$ on the opposite side, and $T$ on a third side. 
\item[$\sigma_3$]  has points $B$, $(1,0)$, $(1,1)$,
  and $T$ appearing in this order, each on a different side of
  $\sigma_3$.
\end{description}
\begin{figure}[ht]
\centering
\includegraphics
{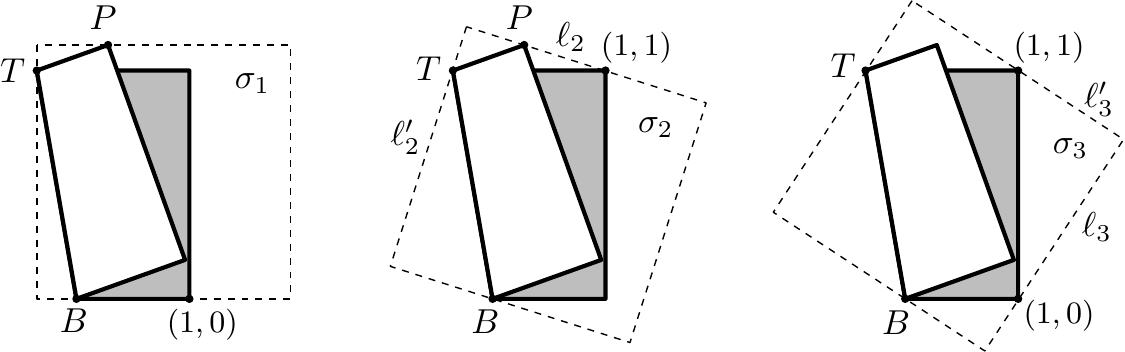}
\caption{Three minimum enclosing squares for $F$.}
\label{fig:square-lower-bound-candidates}
\end{figure}
\begin{theorem}\label{lem:square-upper-bound}
  The $1$-fold cover factor of a square, $c_1^*(\square)$, is
  $\varrho\approx 1.105224$, where $\varrho$ is the largest real root
  of the degree twelve polynomial
  $\Phi(x)=40x^{12}+508x^{11}+1071x^{10}+930x^9-265x^8-1464x^7-1450x^6-524x^5+58x^4+76x^3+3x^2-6x-1$.
\end{theorem}
\begin{proof}
  The effort goes into showing that, for each folded shape $F$, one of
  the three enclosing squares $\sigma_i$, $i\in\{1,2,3\}$, as defined
  above, has side length at most~$\varrho$.

  Denote the side length of a square $\sigma$ by $|\sigma|$.  For a
  start it is easy to see that $|\sigma_1|=p_y<7/5$, which provides a
  first upper bound.

  For $\sigma_2$ we have to consider the distance
  $\mathrm{d}(B,\ell_2)$, where $\ell_2$ is the line through $P$ and
  $(1,1)$ and the distances $\mathrm{d}((1,0),\ell_2')$ and
  $\mathrm{d}(Q,\ell_2')$, where $\ell_2'$ is the line orthogonal to
  $\ell_2$ through $T$. Noting that
  \[
  \mathrm{d}(B,\ell_2)=\frac{\left|{p_x}^2p_y+{p_y}^3+{p_x}^2-2p_xp_y-{p_y}^2-p_y+1\right|}{2p_x\sqrt{(p_x-1)^2+(p_y-1)^2}}
  \]
  \[
  \mathrm{d}({(1,0)},\ell_2')=\frac{\left|{p_y}^2p_x+{p_x}^3-{p_y}^2-3{p_x}^2+2p_y+p_x-1\right|}{2p_x\sqrt{(p_x-1)^2+(p_y-1)^2}},
  \]
  it can be checked that the former dominates the latter for
  $p_y\le\frac{1}{2}(1+\sqrt{4p_x-4{p_x}^2+1})$ and that
  $\mathrm{d}((1,0),\ell_2')>p_y$ for
  $\frac{1}{2}(1+\sqrt{4p_x-4{p_x}^2+1})<p_y<\sqrt{2p_x-{p_x}^2+1}$
  (and so $|\sigma_1| \leq |\sigma_2|$ in such a case). Exactly the
  same holds if $\mathrm{d}((1,0),\ell_2')$ is replaced by
  \[
  \mathrm{d}(Q,\ell_2')=\frac{\left|N_1\right|}{2p_x(1+(p_x-p_y)^2)\sqrt{(p_x-1)^2+(p_y-1)^2}},
  \]
  where
  $N_1={p_x}^5+2{p_x}^3{p_y}^2+p_x{p_y}^4-{p_x}^4-2{p_x}^3p_y-2p_x{p_y}^3+{p_y}^4-4{p_x}^2p_y-4{p_y}^3+4{p_x}^2+2p_xp_y+6{p_y}^2-p_x-4p_y+1$. This verifies that $\sigma_2$ is enclosing, with side length $|\sigma_2|=\mathrm{d}(B,\ell_2)$.

  For $\sigma_3$ we consider a line $\ell_3:y=m(x-1)$ through $(1,0)$,
  for some $m>0$ and the orthogonal line $\ell_3':y=(m+1-x)/m$ through
  $(1,1)$. If $\sigma_3$ is a smallest enclosing square, then
  $\mathrm{d}(T,\ell_3)=\mathrm{d}(B,\ell_3')$. For our range of
  parameters, the only solution is 
  \[
  m=\frac{{p_x}^2+{p_y}^2-1}{{p_x}^2+(p_y-1)^2},
  \]
  which yields 
  \[
  |\sigma_3|=\mathrm{d}(T,\ell_3)=\frac{\sqrt{2}|N_2|}{4p_x\sqrt{D_2}},
  \] 
  where
  $N_2={p_x}^4+2{p_x}^2{p_y}^2+{p_y}^4-4{p_x}^3-2{p_x}^2{p_y}-4{p_x}{p_y}^2-2{p_y}^3+4{p_x}p_y+2p_y-1$
  and
  $D_2={p_x}^4+2{p_x}^2{p_y}^2+{p_y}^4-2{p_x}^2p_y-2{p_y}^3+2{p_y}^2-2p_y+1$.

  Because we choose the smallest square among $\sigma_1$, $\sigma_2$,
  and $\sigma_3$, the claim certainly holds for $|\sigma_1|=p_y\le
  \varrho$.

  It can be checked that $|\sigma_2|\le \varrho$, for all $P$ with
  $\varrho<p_y<\sqrt{2p_x-{p_x}^2+1}$, except for a small region
  $\mathcal{R}$. This region $\mathcal{R}$ is bounded from below by
  the line $y=\varrho$ and from above by the curve $\gamma:|\sigma_2|=\varrho$
  (the branch of this curve that lies in $\{(x,y):\varrho\le
  y<\frac{1}{2}(1+\sqrt{4x-4x^2+1})\}$). The curve $\gamma$
  intersects the line $y=\varrho$ at two points, whose $x$-coordinates are
  approximately $0.67969$ and $0.77126$, respectively. The more
  interesting of these two is the first point of intersection, which
  can be described exactly as the smallest positive real root $x_\varrho$ of
  the polynomial
  $40x^{12}-116x^{11}-1045x^{10}+4756x^9-10244x^8+7260x^7
  -8392x^6-184x^5+620x^4-160x^3+1088x^2-192x+256$.  For the fold
  defined by $P=(x_\varrho,\varrho)$ we have $|\sigma_1|=|\sigma_2|=|\sigma_3|=\varrho$,
  while for all other points in $\mathcal{R}$ the corresponding value
  for $|\sigma_3|$ is strictly less than $\varrho$.

  It can also be checked that $|\sigma_3|<\varrho$, for any $P$ with $p_y>\varrho$
  and
  $\frac{1}{2}(1+\sqrt{4p_x-4{p_x}^2+1})<p_y<\sqrt{2p_x-{p_x}^2+1}$
  (above we committed to using $\sigma_2$ only if $p_y\le
  \frac{1}{2}(1+\sqrt{4p_x-4{p_x}^2+1})$).

  Altogether it follows that $\min\{|\sigma_i|:i\in\{1,2,3\}\}\le
  \varrho\approx 1.105224446$, as claimed. 

  For the other direction, consider the fold defined by
  $P=(x_\varrho,\varrho)$, for which
  $|\sigma_1|=|\sigma_2|=|\sigma_3|=\varrho$. Using
  \lref{square-boundary}, it is easy to check that
  $\sigma_1,\sigma_2,\sigma_3$ are exactly the minimum enclosing
  squares for $F$.
\end{proof}
\subsection{Polygons}\label{app:pgon}
In contrast to disks, polygons can always be made larger with a single
fold; that is, $c_1^*(P)>1$ for all shapes  bounded by a finite
cyclic sequence of vertices and edges with no self-intersections.

\begin{lemma}
For every plane polygon $P$, the 1-fold cover factor $c_1^*(P)>1$.
\end{lemma}
\begin{proof}
We look for finite sets of structures in $P$ that, if not destroyed by
folding, can only be covered by members of that set.  For example, the
set of diameters in a polygon is finite because the maximum distance
$D$ is realized only by pairs of vertices, and any diametral pair that
remains at distance $D$ in the folded state~$F$ must be covered by a
diameter of $P$, possibly itself.

We proceed through increasingly elaborate structures as we consider
different cases for the polygon~$P$.
To begin simply, suppose that in $P$ there exist vertices that
participate in two or more diametral pairs.  (E.g., for odd $n$,
every vertex of a regular $n$-gon.)  Choose as our structure two
diametral pairs, $pq$ and $qr$, that minimize $\theta=\angle pqr$.
Fold along a line trisecting $\theta$, reflecting $qr$ to create~$qr'$
in the folded shape~$F$.  This modified structure has angle 
$\angle pqr'=\theta/3$ between two diameters; by
minimality of $\theta$, it cannot be covered by $P$.

Normalize so that the diameter length $D=1$.  Choose $0<\epsilon<1/4$
so all edge lengths in $P$ are at least $2\epsilon$ and the distance
between any non-diametral pair of vertices is at most $1-\epsilon$.
Define a {\it double arrow} to be a diameter $pq$, plus
$\epsilon$-length segments, $ap$, $bp$, $cq$ and $dq$, of the polygon
edges incident on $p$ and $q$, respectively.  Denote the positive angles these four edge
form to $pq$ by $\alpha$, $\beta$, $\gamma$, and $\delta$,
respectively.  Assume, by relabeling and reflecting, if necessary,
that $\alpha+\beta \ge \gamma+\delta$ and $\gamma\ge \delta$, as
illustrated in \fref{pgonfolds}.  Note that angle $\angle paq$ is
obtuse because $a$ is at most the midpoint of its edge and the distance
$pq$ is at least the distance to the other endpoint.  Similarly, the
other three edges also form obtuse triangles with $pq$.

\begin{figure*}[ht] 
\centering
\includegraphics[width=.9\textwidth]{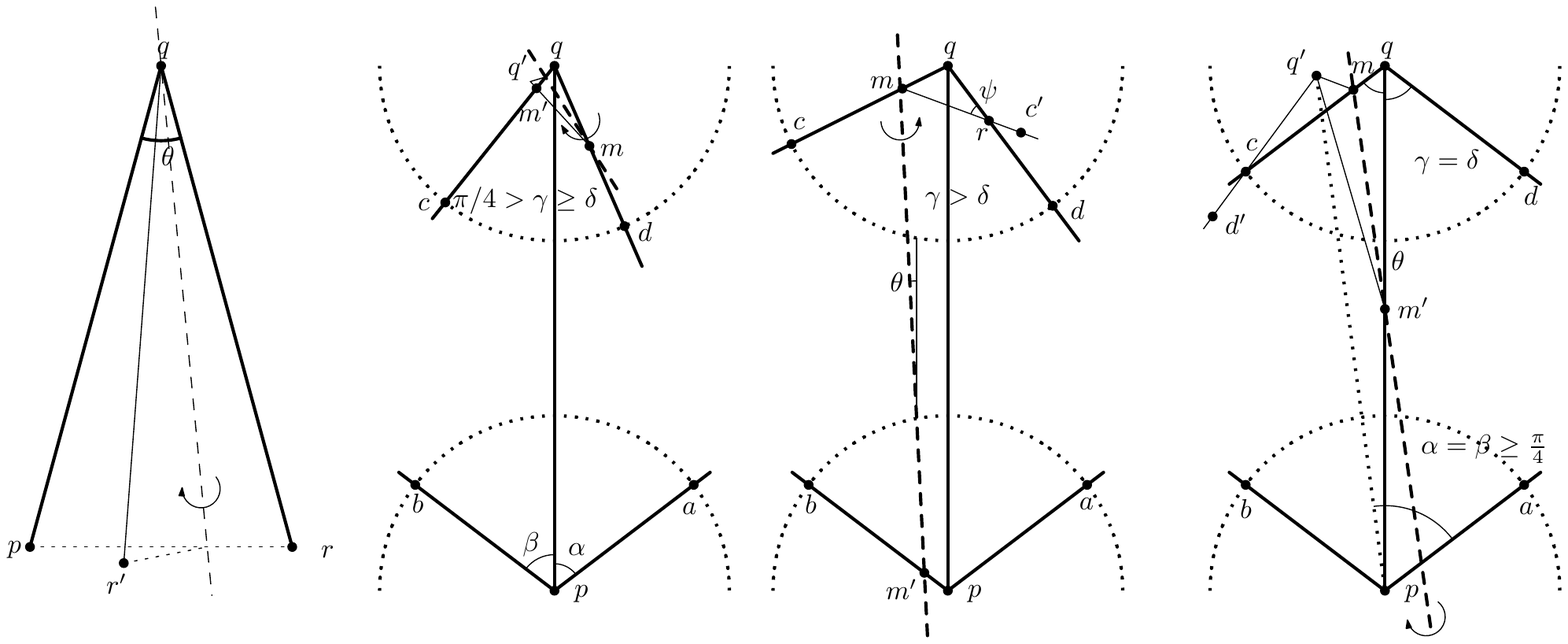}
\caption{Structures in polygons and the folds used to show $c_1^*(P)>1$.}
\label{fig:pgonfolds}
\end{figure*}

From here on, we restrict attention to the set of {\it max arrows}
that lexicographically maximize the 4-tuple
$(\alpha+\beta,\gamma+\delta, \gamma, \max\{\alpha,\beta\})$.  (That
is, when comparing two 4-tuples, compare the pair of $j$th elements
if and only if for all $k<j$ the $k$th elements are equal.)  Not only must max
arrows be covered by max arrows, but we can find other structures
within them.

{\it Case} $\gamma+\delta<\pi/2$: For this case we define a set of blunt
  arrow structures by removing the sharper tip of the max arrow.  Let $m$ be
the midpoint of $dq$. Drop a perpendicular from $m$ to $cq$ at $m'$. A
{\it blunt arrow} is any isometric copy of the heptagon $apbcm'md$. 
Because $mp>1-\epsilon$, a blunt arrow can  be covered only near a
diameter, and because it includes portions of all four edges, these must be 
covered by a max arrow (possibly with reflections and rotations if enough 
of the angles have equal values.)  Thus, a blunt arrow that is not destroyed by folding
must be covered by one of a finite set of blunt arrows.  

Folding the arrow tip $q$ along the bisector of $\angle m'mq$ does not
destroy the blunt arrow. However, by making $mq'$ extend $mm'$, it does put
$q'$ outside any copy of $P$ that covers this blunt arrow.

{\it Case} $\gamma+\delta\ge \pi/2$ and either $\gamma>\delta$ or
$\alpha\ne\beta$: For this case we define a narrower dart structure.
We have assumed that $\gamma>\delta$. Because we will use the fact that
$\alpha+\beta\ge \gamma+\delta$ only to conclude that $\alpha+\beta\ge
\pi/2$, the cases for $\alpha>\beta$ and $\beta>\alpha$ can be handled
in the same way.

Let $q'=2p-q$ be the point obtained by reflecting $q$ at $p$. Draw a line 
through $q'$ that forms an angle $\theta
=(\gamma-\delta)\frac{\epsilon}{6}$ with $\overline{pq}$, and
intersects $bp$ at $m'$ and $cq$ at $m$. (See the third subfigure of \fref{fig:pgonfolds}.)
A {\it dart} is a copy of
the hexagon $apm'mqd$. Because a dart has a diameter segment and
portions of all four incident edges, it should be clear that darts
that survive folding to appear in $F$ must be covered by darts in $P$.

The angle $\theta$ is chosen so that if we fold along the line $\overline{mm'}$, the
reflected $mc'$ intersects $qd$ and extends outside the dart and the polygon.
This claim can be proven as follows. For 
$0<\theta<\pi/4$, we can bound $\frac{2\sqrt2 }{\pi}\theta< \sin\theta<\theta$.  
Let $r=mc'\cap qd$
and apply the law of sines to $\triangle mrq$ to show that $r$ is on both segments $mc'$ and $qd$: 
\begin{align*}
 \angle mrq &
 =(\gamma-\delta)+2\theta>(\gamma-\delta).\\ 
 mq&=\frac{2\sin\theta}{\sin(\gamma+\theta)}<2\sqrt 2\theta
 =\epsilon\frac{\sqrt 2(\gamma-\delta)}{3}.
\end{align*}

It follows that
\[
 qr<mr+mq=mq\Bigl(\frac{\sin(\gamma+\delta)}{\sin (\angle mrq)}+1\Bigr)
 <\epsilon\frac{\sqrt 2(\gamma-\delta)}{3}\Bigl(\frac{\pi}{2\sqrt2(\gamma-\delta)}+1\Bigr)
 <
 \epsilon\Bigl(\frac{\pi}{6}+\frac{\sqrt2}{3}\frac{\pi}{4}\Bigr)<\epsilon.
\]

\smallbreak{\it Case} $\alpha=\beta\ge\pi/4$ and $\gamma=\delta\ge\pi/4$: 
Because of $\alpha+\beta\geq\gamma+\delta$, we have $\alpha\ge \gamma$. 

For this case we define a family of bent arrows, parametrized by angle
$\theta$, which is the angle the fold line $\ell_\theta$ makes with
$pq$.  We show that for sufficiently small $\theta$, a triangle with side
on $pa$ and vertex $q'$ cannot be covered by a copy of $P$.

  To complete the
specification of the fold line $\ell_\theta$, choose its point of
intersection $m'=\ell_{\theta}\cap \overline {pq}$ at distance $h=\epsilon
\cos\theta/\cos(\gamma-\theta)$ from $q$.  This makes the folded image
$q'd'$ pass through~$c$, placing $q'$ outside of the arrow.  Folding
in the other direction, we see that points in the neighborhood of $c'$
go outside of $qd$.  Thus, neither $P$, nor the reflection of $P$ can
cover the folded shape $F$ by aligning $pq$ to a diameter. On the
other hand, $pq$ must remain in the neighborhood of a diameter, because
$|pq'|\ge D-\epsilon$.  (Recall that diameter length $D=1$.)

If $\epsilon > \cos\gamma$, then we can
choose $\theta>0$ sufficiently small that $h>1$ so folding misses the diameter
$pq$ and forms a dart, as in the previous case.  So assume that
the fold line $\ell_{\theta}$ crosses the segment $pq$.

If we calculate the position of $q'=m'+(h,2\theta)$, we observe that  the
derivative with respect to $\theta$ is perpendicular to $pq$: 
$$\frac{d}{d\theta} q'=\frac{2\epsilon}{\cos\gamma}(q-p)^\bot.$$ 
 
Let $e$ be a portion of the polygon edge giving $pa$, truncated to
length $2\epsilon$.  (Thus, $a$ is the midpoint of~$e$.) We consider
covering the triangle with side $e$ and vertex~$q'$.  Initially,
assume that we try to do so by rotating an arrow with maximal angle~$\alpha$.  
Let $C$ be the circle that is tangent to $pb$ at $p$ and has $e$ as a chord. Let $o\ne p$ be the other point intersection $C\cap pq$.

Now, fix edge $e$ in the plane and rotate polygon $P$, keeping $e$ covered by keeping the endpoints of $e$ in contact
with the wedge $bpa$ of angle $2\alpha$.  Point $\hat p$ moves on an
arc of circle $C$ and $\hat q$ moves so that the unit segment $\hat p\hat q$
always passes through $o$.  We can represent this curve (a portion of
a lima\c{c}on) in polar form about the origin $o$, with $\psi$ being
the angle from $oq$, and compute its derivative, which is not
perpendicular to $pq$:
\begin{align*}
\hat q&=(1-2\epsilon\frac{\sin(\psi+\alpha)}{\sin(2\alpha)},\psi)\\
\frac{d}{d\psi} \hat
q&=\frac{\epsilon}{\cos\alpha}(q-p)^\bot-\frac{\epsilon}{\sin\alpha}(q-p).
\end{align*}

This means that we can choose a small $\theta>0$ so that, even in the
widest arrow of angle $2\alpha$, the edge $e$ forces $q'$  outside
the curve traced by $\hat q$.  For narrower arrows, $q'$ is
forced even further outside.  Thus, $P$ cannot cover $F$ without scaling. 
 
This completes the proof that $c_1^*(P)>1$ for any polygon~$P$.
\end{proof}


\subsection{Disks and Lima\c{c}ons}\label{sec:disk1}

\begin{figure*}[ht] 
\centering
\includegraphics[width=0.9\textwidth]{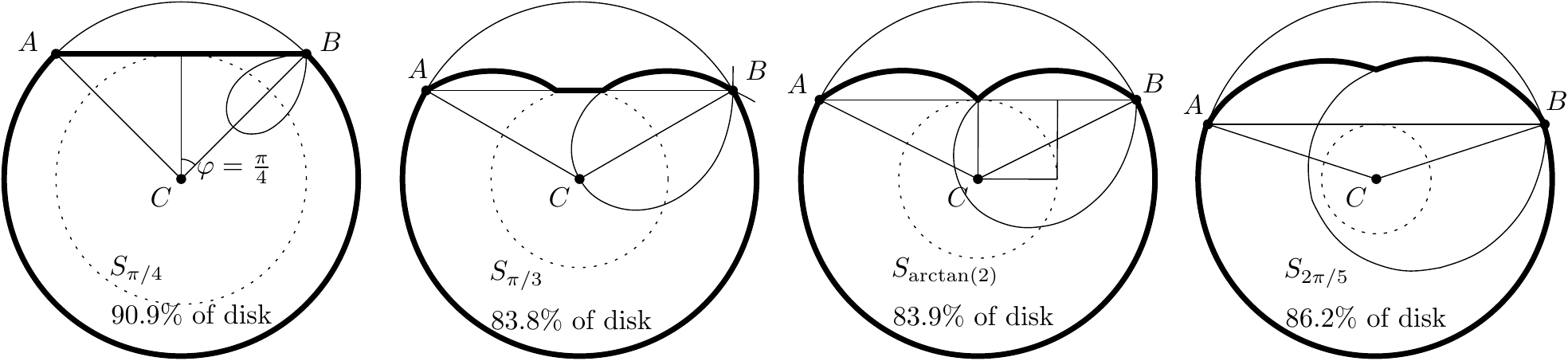}
\caption{Example shapes $S_\varphi$ with $c_1^*(S_\varphi)=1$.  
The dotted circle of radius $\cos \varphi$ gives the drawn loops of
lima\c{c}ons in two ways: as reflections of $B$ across tangents that
separate $B$ and $C$ and as the envelope of circles through $B$ that
are centered on the arc of the dotted circle that lies between the tangents through~$B$. }
\label{fig:disk}
\end{figure*}

A circle sector, folded toward the center, never leaves its circle. This makes it easy to observe that $c_1^*(\bigcirc)=1$.  

We extend this observation to create a family of shapes, parameterized by an angle $\varphi\in
(0,\pi/2)$, that has 1-fold cover factor $c_1^*(S_\varphi)=1$, but
origami cover factor $c^*(S_\varphi)>1$. \fref{disk} shows four
examples. We remove
a circle segment  defined
by a chord $AB$ at distance~$\cos\varphi$ from the disk center~$C$. We
then add back the area bounded by small loops of lima\c{c}ons traced by
reflections of $A$ and $B$ across lines at distance~$\cos\varphi$
from~$C$.  For $\varphi\le\pi/4$,
the lima\c{c}ons do not appear on the boundary of $S_\varphi$, and
for $\varphi> \arctan(2)\approx 63.4^\circ$, the interior of $AB$ does
not appear.

First, we show that the reflection traces a lima\c{c}on, usually defined by tracing a point as one circle is rotated about another of the same radius.  Its standard expression in polar coordinates is $r(\theta)=a+b\cos\theta$.  
Let $\cal C$ denote the circle of radius $\cos\varphi$ centered at~$C$.

\begin{lemma}\label{limacon} 
The loci of the reflections of $A$ about all lines  for  which circle $\cal C$ is contained in one closed halfplane and point~$A$ in the other closed halfplane
 is  the inner loop of a lima\c{c}on.  
\end{lemma}
\begin{proof}
For this proof only, choose a coordinate system with origin $A$, scaled so $\cal C$ is a unit circle with center at  
$C=(-d,0)$ with $d>1$.
The boundary of the loci consists of the  reflections of $A$ about lines tangent to $\cal C$, which can be parametrized by $\theta$.  Let $v=(\cos\theta,\sin\theta)$.  
The  reflection of $A$ about the line tangent to $\cal C$ at $C+v$ is
 $A-2\bigl((A-C-v)\cdot v\bigr)v$, which can be written in polar form as $r(\theta)=2-2d\cos\theta$.  
To see that we get the inner loop, notice that
the tangents to the circle $\cal C$ through $A$  are the extreme lines that satisfy the hypothesis, and these reflect $A$ to itself.
\end{proof}
With this formula, one can verify that the distance from $C$ is unimodal, with the maximum at $A$ and minimum at the other intersection of the loop with $CA$.

Another standard characterization of this loop of the lima\c{c}on\footnote{A demo by Daniel Joseph: \url{http://demonstrations.wolfram.com/LimaconsAsEnvelopesOfCircles/}} is as the intersection of disks centered on~$\cal C$ and containing~$A$.  We know how disks fold, so:
\begin{lemma}\label{limaconfold}
Let $\ell$ be any line that does not intersect the arc of $\cal C$ that lies between the two tangents through~$A$.  
The image of the lima\c{c}on loop by folding across  $\ell$ is contained in the loop.
\end{lemma}
\begin{proof}
Let $\ell^+$ denote the closed halfplane of $\ell$ that contains the arc, and thus contains all centers of disks used to define the lima\c{c}on loop $\cal L$.  
For any individual disk $O$ we fold a sector toward the center; the image of folding $O\setminus \ell^+$ across $\ell$ remains inside $O\cap\ell^+$. Since this is true for all disks, the portion of the loop ${\cal L}\setminus\ell^+$ folds inside the intersection~${\cal L}\cap\ell^+$. 
\end{proof}

Now we can determine the 1-fold cover factor for $S_\varphi$.
\begin{theorem}
The shape $S_\varphi$, with $\varphi\in (0,\pi/2)$,
has 1-fold cover factor $c_1^*(S_\varphi)=1$.
\end{theorem}
\begin{proof}
For a given line $\ell$, let $\ell^+$ be the closed halfplane containing~$C$. 
Construct a folded shape $F$ by reflecting
across line~$\ell$ one or more components of $S_\varphi\setminus \ell^+$.  
We consider cases for the fold
line~$\ell$ based on its distance from~$C$ and the types of boundary
curves of~$S_\varphi$ that it intersects.  

Suppose first that the fold line $\ell$ is less than $\cos\varphi$
from~$C$ (i.e., intersects a dotted circle in \fref{disk}) and that all
components of $S_\varphi\setminus \ell^+$ are folded over.   
Fold the entire 
unit disk along~$\ell$ then rotate about $C$ to make $\ell$ parallel to
$AB$. Since this folded unit disk is covered by $S_\varphi$, the subset
$F$ is certainly covered by~$S_\varphi$.

We sketch the arguments that, for all remaining cases, $F$ remains inside $S_\varphi$
with no rotation needed. It will be enough to consider points of
$S_\varphi\setminus \ell^+$ that fold to locations above $AB$. This can
only occur for $\varphi > \pi/3$, which happens to place $C$ inside the
two lima\c{c}on loops. (This is not crucial to the argument, but does
simplify its geometric interpretation.)

We begin with the cases in which the fold line $\ell$ 
is at least $\cos\varphi$ from~$C$.  Lemma~\ref{limaconfold} implies that 
points in lima\c{c}on loops remain inside their loops after folding. 
So it is enough to consider any point $Q$ on the arc of the unit circle that bounds $S_\varphi$ between $A$ and $B$.  
By Lemma~\ref{limacon}, the image of $Q$ under all possible folds in this case is also a  
 lima\c{c}on loop;  we can rotate this loop around $C$ to align with the loop at $A$ or $B$. 
 
For our chosen fold line $\ell$ we get an image $Q'$ above $AB$. If we extend $CQ'$ through $Q'$ we  hit the loop for $Q$ before we hit the last of  the loops for  $A$ or $B$, thanks to unimodality of distance to $C$.  Thus $Q'$ is inside $S_\varphi$.  

In the remaining case
 $S\setminus \ell^+$ has two components  and we choose to fold only one.  
One component must be a portion of a lima\c{c}on loop that $\ell$ intersects twice above $AB$;  Lemma~\ref{limaconfold} implies that these may be safely folded.  
So we may assume that the fold line $\ell$ has positive slope,  intersects the loop of $B$ twice above $AB$, and intersects the dotted circle $\cal C$ in the upper left quadrant. If 
 $\ell$ crosses the boundary of $S_\phi$ on $AB$, then all points of the folded component go below $AB$.  Otherwise, $\ell$ crosses on the lima\c{c}on loop for $A$. The centers of disks   
that define the portion of this loop above $AB$ are in the lower left quadrant of $\cal C$, so lie in $\ell^+$. As in the proof of Lemma~\ref{limaconfold}, points that fold above $AB$ remain in the intersection of these folded disks, and thus remain in the loop. 
\end{proof}

To see that the origami cover factor
$c^*(S_\varphi)>1$, find the two points of $S_\varphi$ farthest from $(0,-1)$, 
then  crimp fold to narrow the angle between them.
 
The minimum area of $S_\varphi$ is achieved by $\varphi\approx 61.5^\circ$, which occupies about 83.7\% of the unit disk.  The shape is between $S_{\pi/3}$ and $S_{\arctan(2)}$ in \fref{disk}.  


\section{Arbitrary Folds}\label{sec:arb}

\subsection{Simply-Connected Shapes}
\label{sec:simple}

In this section we show that, for a simply connected shape $S$, there is a lower bound
for the origami cover factor $c^*(S)$ that is within a constant factor of the
upper bound given by \lref{ub}.  
\begin{theorem}
\label{thm:simple}
Let $S$ be a simply connected shape with inradius $r$, geodesic radius~$R$, and geodesic diameter~$D$.
Then $\kappa R/r\le D/(2\pi r)\leq c^*(S)\leq R/r$ for
 $\kappa=\sqrt3/(2\pi)\approx 0.27566$. 
\end{theorem}

\begin{proof}
Again, the upper bound is from \lref{ub}. 
The basic idea  for the lower bound is to find a path in~$S$ that can be folded into a
large circle, which 
must then 
be covered by a scaled copy of
the incircle of~$S$.  Here, for brevity, we use a  path of length $D$,
the geodesic diameter. 

\begin{figure}[ht]
\centering
{\includegraphics[width=.7\columnwidth]{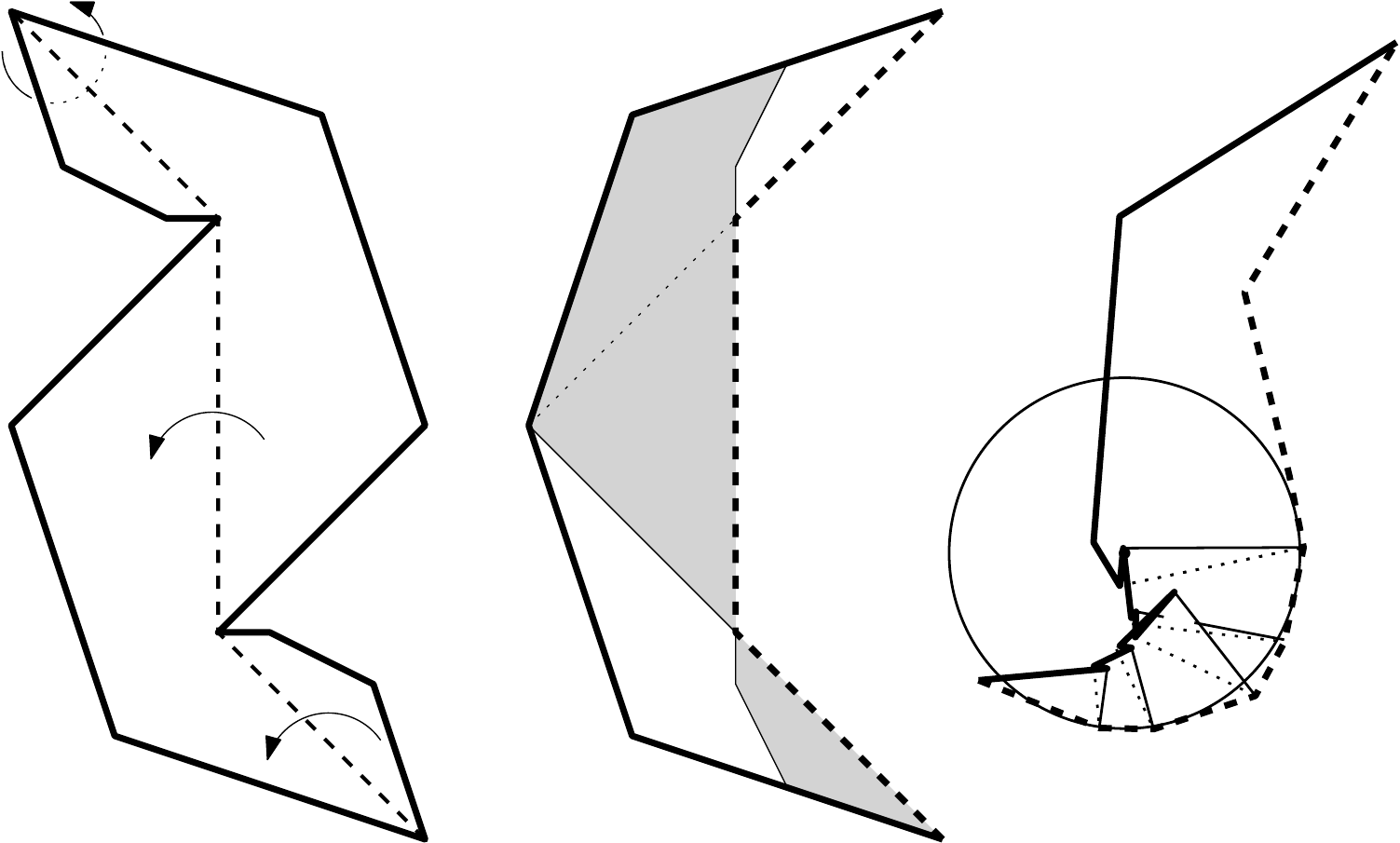}}
\caption{For \protect\tref{simple}, folding inflection edges to make a generalized spiral, then crimping to approximate a circle that must be covered by the incircle.}
\label{fig:spiral}
\end{figure}
A {\it generalized spiral} is a simply connected region composed of 
consistently orientable plane patches having a distinguished shortest
path $\gamma$ that follows the boundary and never turns to the left.  A
generalized spiral may overlap itself if projected onto a plane, but
we can think of it as embedded in a covering space of the plane. 

Ordinarily, a diameter path $\gamma$ will alternate between sequences of
left turns and right turns at boundary points; a portion of the path
between opposite turns is a line segment that we can call an {\it
  inflection edge}.  We can simply fold along every inflection
edge, gluing doubled layers along these edges, to turn $\gamma$ into a
path that goes only straight or to the right.  Folding any
non-boundary edges creates a generalized
spiral with path~$\gamma$.  These folds are along lines of the geodesic
path, so $\gamma$ remains a shortest path between its endpoints.

We  fold the generalized spiral into a left-turning circle
with circumference approaching the length of $\gamma$.  If we sweep a paired 
point and normal vector along $\gamma$, we can think of painting a portion of
the generalized spiral with fibers that each start on $\gamma$ and grow orthogonal to a local tangent (because $\gamma$ is a shortest path) and that are disjoint (because the sweep
in position and angle is monotonic).  We construct a circle whose
circumference is arbitrarily close to the length of $\gamma$ by crimp
folds that align successive fibers of $\gamma$ with the circle
center. \fref{spiral} shows an example.  It does not matter how far
the fibers extend towards or beyond the circle; in order to cover the
boundary of the circle, the inradius~$r$ must be scaled to the circle
radius, which is $D/(2\pi)$.    
\end{proof}



\subsection{Disks with Bumps}\label{sec:disk}

Because the radius of a disk is simultaneously the inradius and the
geodesic radius, \lref{ub} implies that the cover factor of a disk,
$c^*(\bigcirc)$, is $1$. 
It is interesting to note that 
there are other shapes $S$ with
$c^*(S)=1$; here is one simple family.   

In a unit disk centered at $C$ with a chord $AB$, choose a point~$D$
between $C$ and the midpoint of~$AB$.  Add the disk centered at $D$ of
radius~$|AD|$.  Thus, we have a family of shapes $S_{d,e}$,
parameterized by two distances, $d=|CD|$ and 
$e=$ distance from $C$ to chord~$AB$, satisfying 
 $0< d\le e< 1$. See~\fref{bumps}.

\begin{figure}[ht] 
\centering
\includegraphics
{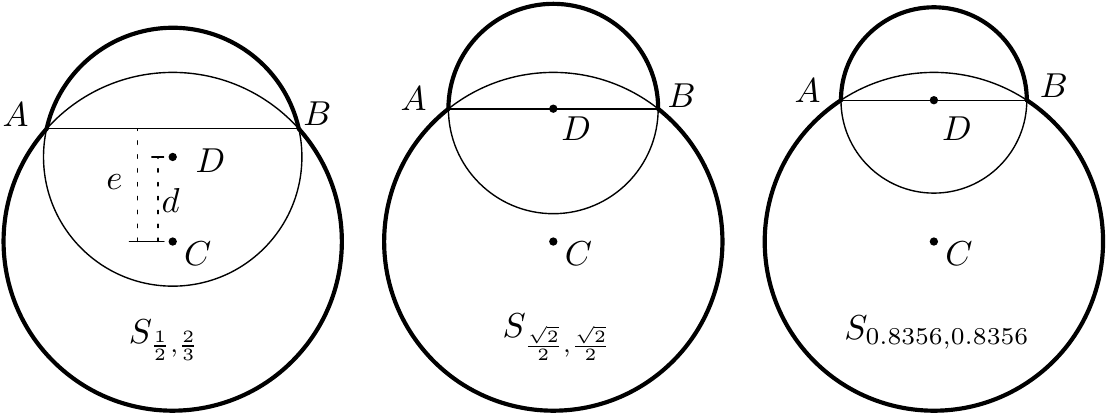}
\caption{Shapes $S_{d,e}$ with $c^*(S_{d,e})=1$.  }
\label{fig:bumps}
\end{figure}

\begin{lemma}\label{lem:disk}
The shape $S_{d,e}$, with $0< d\le e< 1$,
has origami cover factor $c^*(S_{d,e})=1$.
\end{lemma}
\begin{proof} 
Shape $S_{d,e}$ is the union of a unit disk centered at $C$ and 
a disk centered at $D$ whose radius we denote~$r$. Note that by
construction the boundaries of the disks intersect at $A$ and $B$.
This shape also covers all disks of radius $r$ that are centered between
$C$ and $D$.  

Now, in an
arbitrary folded state $S'_{d,e}$, consider the locations of these centers, $C'$
and $D'$.  Placing a unit disk centered at $C'$ and a radius $r$ disk
centered at $D'$ will cover all points of $S'_{d,e}$.  Because
$|C'D'|\le |CD|$, this pair of disks will be covered by placing a
copy of $S_{d,e}$ with $C$ at $C'$ and $D$ on the ray $\ray{C'D'}$.
\end{proof}

Choose any $d\in(0,1)$ and for all $e\in [d,1)$ shape $S_{d,d}$ covers
$S_{d,e}$, so these extremal members of the family have $AB$ as the
diameter of the smaller disk. Just for the sake of curiosity, the example with $d=e=\sqrt2/2$ minimizes
the ratio of inradius to circumradius,
$R/r=(1+\sin\theta+\cos\theta)/2\approx0.8284$, and the example with $d\approx
0.8356$ minimizes the fraction of the circumcircle covered,
$(\pi(1+\sin^2\theta)+\sin2\theta -\theta)/(\pi R^2)\approx 0.7819$.



\section{Open Problems}

The most interesting questions are whether $c^*(\triangle)=c_1^*(\triangle)$ and $c^*(\square)=c_1^*(\square)$, and whether we can completely characterize those shapes with origami or 1-fold cover factor of unity.

\section*{Acknowledgments}
This work began at the 28th Bellairs
Workshop
, March 22-29, 2013.  We thank all other participants for the
productive and positive atmosphere, in particular Godfried Toussaint
for co-organizing the event. We also thank the anonymous 
reviewers for detailed and helpful comments.

{\nocite{mnp-gg-90}
\bibliographystyle{abbrv}
\bibliography{cover}

}
\end{document}